\documentclass[11pt,noinfoline]{imsart}

\RequirePackage[OT1]{fontenc}
\RequirePackage{amsthm,amsmath}
\RequirePackage[numbers]{natbib}
\RequirePackage[colorlinks,citecolor=blue,urlcolor=blue]{hyperref}

\usepackage{verbatim}
\usepackage{bbold}
\usepackage{float}

\allowdisplaybreaks

\usepackage{fullpage}
\usepackage{bbm}

\usepackage{layout}

\usepackage{dsfont}
\usepackage[mathscr]{eucal}
\usepackage[toc,page]{appendix}
\usepackage{mathrsfs}
\usepackage{color}
\usepackage{pifont}
\usepackage{bm}
\usepackage{latexsym}
\usepackage{amsfonts}
\usepackage{amssymb}
\usepackage{epsfig}
\usepackage{graphicx}
\usepackage{multirow}

\usepackage{caption}
\usepackage{subcaption}

\newtheorem{theorem}{Theorem}[section]

\newtheorem{lemma}[theorem]{Lemma}

\newtheorem{proposition}[theorem]{Proposition}
\newtheorem{remark}[theorem]{Remark}

\newtheorem{assumption}[theorem]{Assumption}

\usepackage{verbatim}

\usepackage{tikz}
\usetikzlibrary{fit,positioning,arrows,automata,calc}
\tikzset{
main/.style={circle, minimum size = 5mm, thick, draw =black!80, node distance = 10mm},
connect/.style={-latex, thick},
box/.style={rectangle, draw=black!100}
}

\usepackage{mathtools}
\usepackage{tabu}

\hypersetup{
colorlinks=true,       
linkcolor=blue,        
citecolor=blue,        
filecolor=magenta,     
urlcolor=blue        
}
\newcommand*\diff{\mathop{}\!\mathrm{d}}

\newcommand{\E}{\mathbb{E}}

\newcommand\norm[1]{\left\lVert#1\right\rVert}

\DeclareMathOperator*{\argmin}{argmin}

\begin{document}

\begin{frontmatter}

\title{{\large {Distribution-on-Distribution Regression via Optimal Transport Maps}}}

\runtitle{Distribution to Distribution Regression via Optimal Transport Maps}

\begin{aug}
\author{\fnms{Laya} \snm{Ghodrati}\ead[label=e1]{laya.ghodrati@epfl.ch}} \and
\author{\fnms{Victor M.} \snm{Panaretos}\ead[label=e2]{victor.panaretos@epfl.ch}}


\runauthor{L. Ghodrati \& V.M. Panaretos}

\affiliation{Ecole Polytechnique F\'ed\'erale de Lausanne}

\address{Institut de Math\'ematiques\\
Ecole Polytechnique F\'ed\'erale de Lausanne\\
\printead{e1}, \printead*{e2}}

\end{aug}

\begin{abstract}
We present a framework for performing regression when both covariate and response are probability distributions on a compact interval. Our regression model is based on the theory of optimal transportation and links the conditional Fr\'echet mean of the response to the covariate via an optimal transport map. We define a Fr\'echet-least-squares estimator of this regression map, and establish its consistency and rate of convergence to the true map, under both full and partial observation of the regression pairs. Computation of the estimator is shown to reduce to a standard convex optimisation problem, and thus our regression model can be implemented with ease. We illustrate our methodology using real and simulated data.
\end{abstract}



\end{frontmatter}

{{ \footnotesize
\tableofcontents
}}

\section{Introduction}

Functional data analysis \citep{hsing2015theoretical} considers statistical inference problems whose sample and parameter spaces constitute function spaces. This framework encompasses data that are best viewed as realisations of random processes, and presents challenges arising from the infinite dimensionality of the function spaces, typically taken to be separable Hilbert spaces. On the other hand, non-Euclidean statistics \citep{patrangenaru2015nonparametric} treats inference problems whose sample and parameter spaces are finite dimensional manifolds. Such problems present with a different set of challenges, linked with the non-linearity of the corresponding spaces, which often arises due to non-linear constraints satisfied by the data/parameters.

When the data/parameters of interest are, in fact, probability distributions, one has a problem that is simultaneously functional and non-Euclidean: on the one hand the data can be seen as random processes, and on the other they satisfy non-linear constraints, such as positivity and integral constraints. Thus, the functional data analysis of probability distributions features interesting challenges stemming from this dual nature of the ambient space, for example the finite measurement of intrinsically infinite dimensional objects, and the lack of a linear structure which is crucial to basic statistical operations, such as averaging or, more generally, regression toward a mean. See \citet{petersen-review} for an excellent overview.

One approach to dealing with the non-linear nature of probability distributions is to apply a suitable transformation and map the problem back to a space with a linear structure \citep{kneip2001inference,delicado2011dimensionality,petersen2016functional,kokoszka2019forecasting}. A seemingly more natural approach is to embrace the intrinsic non-linearity, and to analyse the data in their native space, equipped with a canonical metric structure. In the case of probability distributions, the Wasserstein metric \citep{panaretos2019statistical,panaretos2020invitation} has been exhibited as a canonical choice \citep{panaretos2016amplitude}, primarily because it captures deformations, which are typically the main form of variation for probability distributions.

The case of inferring the Fr\'echet mean of a collection of random elements in the Wasserstein space is by now well understood \citep{panaretos2016amplitude,bigot2018upper,zemel2019frechet,gouic2019fast}. The deep links to convexity and the tangent space structure of the Wasserstein space play an important role in motivating and deriving the analysis of this case. The next step is to understand the notion of regression of one probability distribution on another. The first to do so were \citet{chen2020wasserstein}, and, independently, \citet{zhang2020wasserstein}, the latter paper focussing on autoregression. They used the tangent space structure to define a regression operation: using the log transform, the regressor and response are lifted to suitable tangent spaces, where a (linear) regression model is defined in a more familiar Hilbertian setting \citep{morris2015functional,hall2007methodology}. This allows the authors to use the well-developed toolbox of functional regression, and derive appropriate asymptotic theory. 

In this paper, we propose an alternative notion of distribution-on-distribution regression, following a different path. Rather than taking a geometrical approach, via the tangent bundle structure, we follow a shape-constraint approach, namely exploiting convexity. Our model is defined directly at the level of the probability distributions, and stipulates that the response distributions are related to the covariate distributions by means of an optimal transport map, and further deformational noise. A key advantage of this approach is its clean and transparent interpretation, since the regression operator can be interpreted \emph{pointwise} at the level of the original distributions, and its effect consists in mass transportation, or equivalently, quantile re-arrangement. Further to this, the approach requires minimal regularity conditions, and does not suffer from ill-posedness issues as inverse problems do. Finally, its computational implementation reduces to a standard convex optimisation problem. The usefulness of the approach is exhibited when revisiting the analysis of the mortality data of \citet{chen2020wasserstein}, where the approach is seen to lead to similar (if more expansive) qualitative conclusions, but with the advantage of an arguably improved interpretability.

\section{Background on Optimal Transport and Some Notation}{\label{Wasserstein}}
In order to define our regression model, we now provide some minimal background on optimal transport and Wasserstein distances, including some relevant notation. For more background see, e.g. \cite{panaretos2020invitation}. Let $\Omega\subseteq\mathbb{R}$ and $\mathcal{W}_2(\Omega)$ be the set of Borel probability measures on $\Omega$, with finite second moment. The 2-Wasserstein distance $W$ between $\mu,\nu \in \mathcal{W}_2(\Omega)$ is defined by 
$$d^2_{\mathcal{W}}(\nu,\mu):=\underset{\gamma \in \Gamma(\nu,\mu)}{\inf} \int_{\Omega} |x-y|^2 \diff \gamma(x,y),$$
where $\Gamma(\nu,\mu)$ is the set of couplings of $\mu$ and $\nu$, i.e. the set of Borel probability measures on $\Omega \times \Omega$ with marginals $\nu$ and $\mu$. It can be shown that $\mathcal{W}_2(\Omega)$ endowed with $d^2_{\mathcal{W}}$ is a metric space, which we simply call the Wasserstein space of distributions. A coupling $\gamma$ is deterministic if it is the joint distribution of $\{X, T(X)\}$ for some deterministic map $T: \Omega \to \Omega$, called an optimal transport map. In such a case, we write $\nu=T\#\mu$ and say that $T$ pushes $\mu$ forward to $\nu$, i.e. $\nu(B)=\mu\{T^{-1}(B)\}$ for any Borel set $B$. Occasionally we denote this as $T_{\mu\to\nu}$, for clarity.

When the source distribution $\mu$ is absolutely continuous with respect to the Lebesgue measure,  then the optimal plan is induced by a map $T$. When $d=1$, the map $T$ admits the explicit expression $T=F^{-1}_{\nu}\circ F_\mu$, where $F^{-1}_\nu$ is the quantile function of $\nu$, and $F_\mu$ is the cumulative distribution function of $\mu$. In addition

\begin{equation}\label{w1d}
  d_{\mathcal{W}}^2(\mu,\nu)=\int_0^1\big|F^{-1}_{\mu}(p) - F^{-1}_{\nu}(p)\big|^2 \diff p.
\end{equation}

A notion of average of probability distributions can be defined via the Fr\'echet mean with respect to the Wasserstein metric. Namely, let $\Lambda$ be a random measure on $\mathcal{W}_2(\Omega)$ with law $P$. A Fr\'echet mean of $\Lambda$ is a minimizer of the Fr\'echet functional
$$F(b)=\frac{1}{2}E d^2_{\mathcal{W}}(b,\Lambda)=\frac{1}{2}\int_{\mathcal{W}_2(\Omega)}  d^2_{\mathcal{W}}(b,\lambda)\diff P(\lambda)\quad b \in \mathcal{W}_2(\Omega).$$
The Fr\'echet functional can thus serve as a basis to define a sum-of-squares functional in the context of regression, and this will be done in the next section. 

We will occasionally use the fact that $\mathcal{W}_2(\mathbb{R})$ is flat in that for $\mu,\nu, b \in \mathcal{W}_2(\mathbb{R})$ it holds that
\begin{equation}\label{PC}
  d_{\mathcal{W}}(\mu,\nu)=\norm{T_{b \to \nu} - T_{b \to \mu}}_{L^2(b)},
\end{equation}
whenever the optimal maps involved are well-defined.

Finally, we will use the notation $a \lesssim b$ to indicate that there exists a positive constant $C$ for which $a\leq C b$ holds. The support of a function $f$ will be denoted by $\text{supp}(f)$ . And, for a measure $\mu$, we indicate the $L^p$ norm of a function $f:[0,1]\rightarrow \mathbb{R}$ with respect to $\mu$ as $\norm{f}_{L^p(\mu)}$.

\section{Distribution-on-Distribution Regression}{\label{Distribution-on-Distribution Regression}}

\subsection{Fr\'echet Functionals and Regression Operators}\label{regression_operators}

Let $(\mu,\nu)$ be a pair of random elements in $\mathcal{W}_2(\Omega) \times \mathcal{W}_2(\Omega)$ with joint distribution $P$. Then, similar to a standard nonparametric regression model, we can define a regression operator $\Gamma: \mathcal{W}_2(\Omega) \rightarrow \mathcal{W}_2(\Omega)$ as the minimizer of the conditional Fr\'echet functional, viewed as a function of $\mu$,

$$
  \argmin_b \int_{\mathcal{W}_2(\Omega)} d^2_{\mathcal{W}}(b,\nu)\diff P(\nu\,|\,\mu)=\Gamma(\mu)  
$$
assuming that for any $\mu$, the Fr\'echet mean of the conditional law $P(\cdot \,|\, \mu)$ of $\nu$ given $\mu$ is unique 
, which can be enforced by means of regularity assumptions on the pair $(\mu,\nu)$.

The difference between the above formulation and the standard regression formulation is that we have replaced the notion of expectation with a Wasserstein-Fr\'echet mean, an approach termed as ``Fr\'echet Regression" by  \citet{petersen2019frechet}.  
Postulating a specific form on the regression operator $\Gamma^*$ amounts to defining a certain type of regression model. If $\Gamma$ is left unconstrained, except for possessing some degree of regularity, then we would speak of a nonparametric regression model. However, assumptions on $\Gamma$ are needed to ensure its identifiability, and simply assuming it is regular will not suffice in this more general context. 

For instance, the approach of \citet{chen2020wasserstein} and \citet{zhang2020wasserstein} consists in constraining $\Gamma$ to be in a certain sense linear, in that it can be represented as a linear operator at the level of the tangent bundle. Identifiability, and indeed fitting and asymptotic theory, can then be derived by appealing to the inclusion of the tangent spaces in Hilbert spaces. 

Here we impose a different constraint on $\Gamma$, and consequently define a different notion of regression. Namely we impose a \emph{shape constraint}, by assuming that $\Gamma(\mu)=T\#\mu$, where $T$ is an increasing map. This is developed in the next section which postulates a regression model on the pair $(\mu,\nu)$ that guarantees the uniqueness of the conditional Fr\'echet mean $\Gamma(\mu)$ of $\nu$ given $\mu$, and imposes mild conditions ensuring the identifiability of $\Gamma$.

\subsection{The Regression Model and The Fr\'echet-Least-Squares Estimator}
{\label{The Model and The Estimator}}

Henceforth, we will take the domain $\Omega$ to be a compact interval of $\mathbb{R}$. Let $\{(\mu_i,\nu_i)\}_{i=1}^{N}$ be an independent collection of regressor/response pairs in $\mathcal{W}_2(\Omega)\times \mathcal{W}_2(\Omega)$. Motivated by the discussion in the previous paragraph, we define the regression model
\begin{equation}{\label{model}}
   \nu_{i}=T_{\epsilon_i}\#(T_0\#\mu_i),  \quad  \{\mu_i,\nu_i\}_{i=1}^N,
\end{equation}
where $T_0:\Omega \to \mathbb{R}$ is an unknown optimal map and $\{T_{\epsilon_i}\}_{i=1}^{N}$ is a collection of independent and identically distributed random optimal maps satisfying $E\{T_{\epsilon_i}(x)\}=x$ almost everywhere on $\Omega$. These represent the ``noise" in our model. The regression task will be to estimate the unknown $T_0$ from the observations $\{\mu_i,\nu_i\}_{i=1}^N$. To be able to do so, we need to ensure that $T_0$ is identifiable, and for this we now introduce some conditions.

In the spirit of Section \ref{regression_operators}, let  $P$ be the probability law induced on $\mathcal{W}_2(\Omega) \times \mathcal{W}_2(\Omega)$ by model \eqref{model}. We denote by $P_{M}$ and $P_N$ the marginal distributions induced on the typical regressor $\mu$ and the typical response $\nu$, respectively.

\begin{assumption}{\label{absCont}}
Let $\mu$ be a measure in the support of $P_M$. Then $\mu$ is absolutely continuous with respect to Lebesgue measure on $\Omega$.
\end{assumption}

Denote by $Q$ the measure that is linear average of $P_M$, i.e. $Q(A)=\int_{\mathcal{W}_2(\Omega)} \mu(A)\diff P_M(\mu)$. We also denote by $Q_N$ the empirical counterpart of $Q$, namely $Q_N(A)=\frac{1}{N} \sum_{i=1}^N \mu_i(A)$, where $\{\mu_i\}$ are independent random measures with law $P_M$. Note that all $\mu$ in the support of $P_M$ are dominated by the measure $Q$, i.e. $\mu \ll Q$ almost surely.

Define the parameter set of optimal transport maps $\mathcal{T}$ as:
$$\mathcal{T}:=\{T :\Omega \to \Omega: 0\leq T'(x) {<\infty} \text{ for } Q \text{-almost every } x \in \Omega \}.$$

Implicit in the definition of $\mathcal{T}$ is that its elements are assumed differentiable $Q$-a.e. In the presence of Assumption \ref{absCont}, the $Q$-a.e. existence of $T'$ is automatically guaranteed, since Lebesgue's theorem on the differentiation of monotone functions states that a monotone function automatically has a derivative Lebesgue almost everywhere in the interior of $\Omega$, and Assumption \ref{absCont} implies that $Q$ is dominated by Lebesgue measure.\\ 

\noindent We will also assume:
\begin{assumption}{\label{assumpMaps}}
The model (\ref{model}) is induced by a map $T_0$ and random maps $T_{\epsilon}$ that are of class $\mathcal{T}$.
\end{assumption}

\noindent With these assumptions in place, we can now establish identifiability:

\begin{theorem}{\label{identifiability}}
Assume that the law $P$ induced by model (\ref{model}) satisfies Assumptions \ref{absCont} and \ref{assumpMaps}. Then, the regressor operator $\Gamma(\mu)=T_0\#\mu$ in model \eqref{model} is identifiable over the parameter class $\mathcal{T}$ in the $L^2(Q)$ topology. Specifically,  for any $T\in\mathcal{T}$ such that $\|T-T_0\|_{L^2(Q)}>0$, it holds that
$$M(T)>M(T_0),$$  
where 
\begin{equation}{\label{population-functional}}
 M(T):= \frac{1}{2}\int_{\mathcal{W}_2(\Omega)\times \mathcal{W}_2(\Omega)} d^2_{\mathcal{W}}(T\# \mu,\nu) \diff P(\mu,\nu).
 \end{equation}

\end{theorem}

\begin{remark}[Identifiability $Q$-almost everywhere]
Theorem \ref{identifiability} establishes the identifiability of $T_0$ up to $Q$-null sets, with minimal assumptions on the input measures $\mu$. Consequently, if the random covariate measure $\mu$ is almost surely supported on a strict subset $\Omega_0\subset \Omega$, we can identify $T_0$ on $\Omega_0$ (which coincides with the support of $Q$) but not on $\Omega\setminus\Omega_0$. Of course, if the measure $Q$ is  equivalent to Lebesgue measure, in the sense of mutual absolute continuity, identifiablity will also hold Lebesgue almost everywhere on $\Omega$. Additional conditions on the law of the  random covariate measure $\mu$ can yield this equivalence. A simple condition is to require $\int_{\mathcal{W}_2(\Omega)} \inf_{x\in\Omega} f_\mu(x) \diff P_M(\mu)>0$, yielding that $f_Q(x)>0$, where $f_\mu$ and $f_Q$ are the Lebesgue densities of the measures $\mu$ and $Q$. However this condition implies that $\textrm{supp}(\mu)=\Omega$ with positive probability, which can be restrictive as we would like our model to encompass situations where none of the covariate measures are fully supported on $\Omega$. A considerably weaker condition that guarantees the equivalence of $Q$ to Lebesgue measure is to require the existence of a cover $\{E_m\}_{m\geq 1}$ of $\Omega$ such that $P_M\{E_m\subseteq\mathrm{supp}(f_\mu)\}>0$ for all $m$ -- intuitively, this enables different covariate measures to  give information on $T_0$ on different subsets of $\Omega$, but requires that they collectively provide information on all of $\Omega$. 
As an example let $\Omega=[0,1]$ and let $\mu$ be defined as the normalised Lebesgue measure on $S=[U,U+1/3] \mod 1$, where $U$ is a uniform random variable on $[0,1]$. In this case none of the realisations of $\mu$ are supported on $\Omega$, but the ``cover condition" is satisfied.
\end{remark}

Further to identifiability, the theorem gives a way to estimate $T_0$ by means of $M$-estimation. We can define an estimator $\hat{T}_N$ as the minimizer of the sample counterpart of $M$,
\begin{equation}{\label{functional}}
    M_N(T):= \frac{1}{2N} \sum_{i=1}^N d^2_{\mathcal{W}}(T\# \mu_i,\nu_i),
  \quad \quad \hat{T}_N:=\arg\min_{T \in \mathcal{T}}   M_N(T),
\end{equation}
where $(\mu_i,\nu_i)$ are independent samples from $P$ for $i=1,\dots,N$. In effect this a ``Fr\'echet least square" estimator. The existence and uniqueness of a minimizer is not a priori obvious, but we establish both in the next section \ref{Existence and Uniqueness of the Estimator}.

\begin{remark}[Pure Intercept Model]
When all the input measures are equal, $\mu_1=\hdots=\mu_N$, our regression model reduces to a ``pure intercept model", which is equivalent to the problem of estimating a Fr\'echet mean. To see this, let $\mu_0$ a fixed measure. From the assumption that $E\{T_{\epsilon_i}(x)\}=x$ a.e., one can deduce that the conditional Fr\'echet mean of the measure $\nu$, given the measure $\mu_0$ is equal to $\nu_0=T_0\# \mu_0$. Estimation of $T_0$ is then equivalent to estimation of the Fr\'echet mean $\nu_0$ of the output measures, since $T_0=T_{\mu_0\to\nu_0}=F_{\nu_0}^{-1}\circ F_{\mu_0}$.
\end{remark}

\subsection{Interpretation and Comparison}

It was argued in the introduction that the proposed regression model has the advantage of being easily interpretable, and now we elaborate on this point. The fact that the regressor operator $\Gamma(\mu)$ takes the form 
\begin{equation}\label{our-regressor}
\Gamma(\mu)=T_0\# \mu,
\end{equation}
 where $T_0:\Omega\to\Omega$ is a monotone map, has a simple interpretation in terms of mass transport: the effect of the Fr\'echet mean in this regression is to transport the probability mass assigned by $\mu$ on a subinterval $(a,b)\subset \Omega$ onto the transformed subinterval $(T_0(a),T_0(b))$. Therefore, the model can be directly interpreted at the level of the quantity that the input/output measures are modelling. In particular, the model can be interpreted at the level of quantiles. Since 
$$F^{-1}_{T_0\#\mu}(\alpha)=(T_0\circ F^{-1}_{\mu})(\alpha)=T_0\{F^{-1}_{\mu}(\alpha)\},\qquad \alpha\in(0,1),$$
we can see that the mean effect of the regression is to move the $\alpha$-quantile of $\mu$, say $q_\alpha$, to the new location $T_0(q_\alpha)$. Each response distribution $\nu_i$ will then further deviate from its conditional Fr\'echet mean $T_0\#\mu_i$ by means of a random monotone ``error" map $T_{\epsilon}:\Omega\to\Omega$ whose expectation is the identity map,
$$F^{-1}_{\nu_i}(\alpha)=T_{\epsilon_i}\big[T_0\{F^{-1}_{\mu}(\alpha)\}\big],\qquad \alpha\in(0,1).$$
This highlights the analogy with a classical regression setup, except that the addition operation is replaced by the composition operation at the level of quantiles, or equivalently, by the push-forward operation at the level of distributions. In particular, the assumption that $E\{T_{\epsilon_i}(x)\}=x$ is directly analogous to the classical assumption that the errors have zero mean: one can directly see that $E\{T_{\epsilon_i}(x)\}=x$ for almost all $x\in\Omega$ implies that
$$E\{F^{-1}_{\nu_i}(\alpha)\}=E\Big(T_{\epsilon_i}\big [ T_0\{F^{-1}_{\mu}(\alpha)\}\big]\Big)=T_0\{F^{-1}_{\mu}(\alpha)\},\qquad \alpha\in(0,1).$$

Assuming that we have obtained an estimator $\hat T_N$ of the regression map $T$ based on $N$ regressor/response pairs, we can then define the \emph{fitted distributions}, 
$$\hat{\nu}_i=\hat{T}_N\# \mu_i.$$
We can also define the $i$th \emph{residual map} $T_{e_i}(x):\Omega\to\Omega$ as the optimal transport map $T_{e_i}=T_{\nu_i\to\hat\nu_i}$ that pushes forward the observed response $\nu_i$ to the fitted value $\hat\nu_i$. The residual maps can be plotted in a ``residual plot" and contrasted to the identity map, by analogy to the classical regression case. This can help identify outlying observations, and also to appreciate in what manner the fitted values differ from the observe values. In particular, it can reveal in which regions of the support of the measures the model provides a good fit, and where less so. It can also serve to identify clusters of observations whose residuals are similar, suggesting the potential presence of a latent indicator variable, i.e. that separate regressions ought to be fit to different groups of observations. Finally, the residual plot can serve as a diagnostic tool for the validity of the model. Since the residual map $T_{e_i}$ can be seen as a proxy for the latent error map $T_{\epsilon_i}$, deviations of the average of the residual maps from the identity can serve as a means to diagnose departures from the assumed model. Note that, contrary to classical regression, where the residuals sum to zero by construction, the residual maps $T_{e_i}$ are \emph{not} constrained to have mean equal to the identity.

By comparison, \citet{chen2020wasserstein} introduce (linear) regression in Wasserstein space by means of a geometric approach, that is in a sense a linear model between tangent spaces. Namely, for $\bar\mu$ and $\bar\nu$, the Fr\'echet means of the regressor and response measures, they postulate a regressor operator of the form
\begin{equation}\label{muller-regressor}
\Gamma(\mu)=\big\{\mathcal{B}(T_{\bar{\mu}\to \mu}-{I})+I\big\}\# \bar{\nu},
\end{equation}
where $I(x)=x$ is the identity map on $\Omega$, and $\mathcal{B}: L^2(\bar{\mu})\to L^2(\bar{\nu})$ is a bounded linear operator with some assumptions, so that the terms involved be well-defined. Again, linearity guarantees identifiability. The expression appears convoluted, but the geometrical interpretation is simple: $T_{\bar{\mu}\to \mu}-I$ represents the image of $\mu$ under the log map at $\bar{\mu}$ (see Section 2.3 of \citet{panaretos2020invitation}). Equivalently, $T_{\bar{\mu}\to \mu}-I$ is the lifting of $\mu$ to the tangent space $\mathrm{Tan}_{\bar{\mu}}\{\mathcal{W}_2(\Omega)\}\subset L^2(\bar{\mu})$ at $\bar{\mu}$. Once the regressor $\mu$ is lifted onto $\mathrm{Tan}_{\bar{\mu}}\{\mathcal{W}_2(\Omega)\}\subset L^2(\bar{\mu})$, the action of the regression operator is to map it to its image in $L^2(\bar\nu)$ via the bounded linear operator $\mathcal{B}:L^2(\bar{\mu})\to L^2(\bar{\nu})$, as in a standard functional linear model. The final step is to push forward $\bar{\nu}$ by this image plus the identity, i.e. $\mathcal{B}(T_{\bar{\mu}\to \mu}-{I})+I$, which retracts back onto $\mathcal{W}_2(\Omega)$ and yields a measure (if $\mathcal{B}(T_{\bar{\mu}\to \mu}-{I})+I$ is a monotone map, then this is equivalent to exponentiation, see Section 2.3 of \citet{panaretos2020invitation}). The model is most easily interpretable on the tangent space, where it states that the expected lifting of the response $\nu_i$ at $\bar{\nu}$ is related to the lifting of the regressor $\mu_i$ at $\bar{\mu}$ by means of the linear operator $\mathcal{B}$. Similarly, fitted values are defined on the tangent space, and then can be retracted by the same push-forward operation.

The two approaches do not directly compare, and neither captures the other as a special case. Similarly, there is no reason to a priori expect that one model would typically outperform the other in terms of fit, and one can expect this to depend on the specific data set at hand. Thus, our method should be seen as an alternative rather than an attempt at an improved or more general version of regression. An apparent advantage of the regressor function \eqref{our-regressor}, however, is an arguably easier and more direct interpretation of the regression effect, directly at the level regressor/response, through a monotone re-arrangement of probability mass, as discussed above. Indeed this allows a direct point-wise interpretation of the regression effect. The regressor \eqref{muller-regressor} on the other hand allows for a traditional (functional) regression interpretation via the linear operator $\mathcal{B}$, albeit acting on the logarithms of regressor/response, which makes it harder to interpret the regression effect at the level of the original measures, since there are two transformations involved, one non-linear and one linear. Similar points can be made with regards to the residuals and residual plots. Another potential advantage is at the level of regularity conditions imposed on $\Gamma$ for the purposes of theory. Equation \eqref{muller-regressor} leads to an inverse problem on the tangent space, as is standard with functional linear models, and thus requires more delicate technical assumptions on the problem, in addition to regularisation. By contrast, the shape-constrained approach \eqref{our-regressor} only requires monotonicity on the regressor $T_0$. It also avoids the instabilities of an inverse problem.

The utility of our model illustrated in Section \ref{mortality_data}, which considers an example where the age-at-death distribution $\nu_i$ for country $i$ in 2013 serves as a response distribution, and the age-at-death distribution $\mu_i$ of the same country in 1983 serves as the regressor. Interestingly, it leads to similar fits and qualitative conclusions as the analysis of the same data by \citet{chen2020wasserstein}, while exhibiting a clean and more expansive interpretation. Indeed, our definition of residual maps help identify effects related to changes in infant mortality not easily detectable when looking only at the fitted distributions, and to identify an interesting clustering of observations. See Section \ref{mortality_data} for more details.

\subsection{Existence and Uniqueness of the Estimator} {\label{Existence and Uniqueness of the Estimator}}
In this section, we establish the existence and uniqueness of the estimator $\hat{T}_N$. To show the existence, we use a variant of the Weierstrass theorem, namely \citet[Thm 7.3.6]{kurdila2006convex}, stated for convenience as Theorem \ref{theorem for existence} in the Appendix. This requires establishing the convexity and Gateaux differentiability of the functional $M_N$, and this we do in the next lemma:

\begin{lemma}[Strict Convexity and Differentiability]{\label{functional-convexity-derivative}}
Let $\mathcal{T}$ be the parameter set and suppose we have $N$ independent observations $(\mu_i,\nu_i)$ that are realizations of $P$. Both the empirical functional $M_N(T)$ and the population functional $M(T)$ are strictly convex with respect to $T \in \mathcal{T}$.
Moreover the functionals $M$ and $M_N$ are Gateaux-differentiable on the set of optimal maps in $\mathcal{T}$ with respect to the $L^2(Q)$ and $L^2(Q_N)$ distances, respectively. The corresponding derivatives of $M$ in the direction $\eta\in L^2(Q)$ is:
\begin{equation}
D_\eta M(T)=\int\int_{\Omega} \eta(x)\{T(x)-T_{\mu,\nu}(x)\} \diff \mu(x)\diff P(\mu,\nu),  
\end{equation}
and the derivative of $M_N$ in the direction $\eta\in L^2(Q_N)$ is
\begin{equation}
D_\eta M_N(T)=\frac{1}{N}\sum_{i=1}^N \int_{\Omega} \eta(x)\{T(x)-T_{\mu_i,\nu_i}(x)\} \diff \mu_i(x),
\end{equation}
where $T_{\mu,\nu}$ is the optimal map from $\mu$ to $\nu$.
\end{lemma}
\noindent Since $\mathcal{T}$ is a convex, closed, and bounded subset of $L^2(Q)$ functions, we may now apply the Weierstrass theorem cited above to conclude:

\begin{proposition}[Existence and Uniqueness of the Estimator]{\label{Unique-minimizer}}
There exists a unique solution $\hat{T}_N\in\mathcal{T}$ to the Fr\'echet sum-of-squares minimization problem \eqref{functional}, with uniqueness being in the $L^2(Q_N)$ sense.
\end{proposition}

\subsection{Computation}{\label{computation}}

\noindent Since the domain $\Omega$ is one-dimensional, we have that 
$$d^2_{\mathcal{W}}(\nu,\mu)=\int_{0}^1 \big|F^{-1}_\mu(p)-F_\nu^{-1}(p)\big|^2 \diff{p}.$$ 
Furthermore, since the regressors $\mu_i$ are assumed absolutely continuous (Assumption \ref{absCont}), we can always write $\nu_i=T_{\mu_i \to \nu_i}\#\mu_i$ for an optimal map $T_{\mu_i \to \nu_i}$. We can therefore manipulate the Fr\'echet sum-of-squares and use a Riemann approximation to write
\begin{align}
\sum_{i=1}^N d^2_{\mathcal{W}}(T\#\mu_i,\nu_i)=\sum_{i=1}^N \norm{T\circ F^{-1}_{\mu_i} - F^{-1}_{\nu_i}}^2_{L^2} &=\sum_{i=1}^N \int_0^1 \big|T\circ F^{-1}_{\mu_i}(p)-F^{-1}_{\nu_i}(p)\big|^2 \diff p \nonumber \\
&= \sum_{i=1}^N \int_0^1 \big|T\circ F^{-1}_{\mu_i}(p)-T_{\mu_i \to \nu_i} \circ F^{-1}_{\mu_i}(p)\big|^2 \diff p \nonumber \\
&=\sum_{i=1}^N \int_{\Omega} \big|T(x)-T_{\mu_i \to \nu_i}(x)\big|^2 \diff \mu_i (x) \nonumber \\
&\approx \sum_{i=1}^N \sum_{j=1}^m \big|T(x_j)-T_{\mu_i \to \nu_i}(x_j)\big|^2 \mu_i(h_j), \label{Riemann}
\end{align}
for $m$ user-defined nodes $\{x_j\}_{j=1}^{m}$ in an interval partition $\{I_j\}_{j=1}^{m}$  of $\Omega$, and $h_j=|I_j|$. Writing $y_{ij}=T_{\mu_i \to \nu_i}(x_j), w_{ij}=\mu_i(h_j)$ and $z_j=T(x_j)$, we reduce the above approximate minimization of the Fr\'echet sum-of-squares to the solution of the following convex optimization problem:

\begin{equation}
\begin{split}
  &  \text{minimise }f(z)=\sum_{i=1}^N \sum_{j=1}^m w_{ij} h_i(y_{ij},z_j)\\
&\text{subject to  } z_1 \leq z_2 \leq \cdots \leq z_m 
\end{split}
\end{equation}
where $h_i(y_{ij},z_j)=|y_{ij}-z_j|^2$.  The above problem resembles an isotonic regression problem with repeated measurements, and can be solved via the Pool-Adjacent-Violater-Algorithm (PAVA) \citep{mair2009isotone}.

\subsection{Consistency and Rate of Convergence}{\label{consistency and rate of convergence}}

In this section, we establish the asymptotic properties of the proposed estimators both in the case of the fully observed set of measures $\{\mu_i,\nu_i\}$ and the case where one only indirectly observes input/output distributions through i.i.d. samples from each. A natural risk function to measure the quality of the estimator is the Fr\'echet mean squared error:
$$R(T):=\mathbb{E}_{\mu \sim P_M} d^2_{\mathcal{W}}(T_0\# \mu,T\# \mu)=\int_{\mathcal{W}_2(\Omega)} d^2_{\mathcal{W}}(T_0\# \mu,T\#\mu) \diff P_M(\mu).$$
Using the equation \eqref{PC} we can rewrite the above risk as follows:
\begin{equation*}
    \begin{split}
        \int d^2_{\mathcal{W}}(T_0\# \mu,T\#\mu) \diff P_M(\mu) &= \int \norm{T_0-T}_{L^2(\mu)}^2 \diff P_M(\mu) \\
        &=\int \int_{\Omega} \big|T_0(x)-T(x)\big|^2 \diff \mu(x) \diff P_M(\mu) \\
        &=  \norm{T_0-T}^2_{L^2(Q)}
    \end{split}
\end{equation*}

Thus, we can obtain consistency and convergence rates in Fr\'echet mean squared error using the criterion $\|{T_0-\hat{T}_N}\|_{L^2(Q)}$, in particular:

\begin{theorem}\label{rate of convergence}
In the context of model \eqref{model}, suppose that Assumptions \ref{absCont} and \ref{assumpMaps} hold true. Then, the estimator $\hat{T}_N$ defined in \eqref{functional} is a consistent estimator for $T_0$ satisfying
\begin{equation}
N^{1/3} \norm{\hat{T}_{N}-T}_{L^2(Q)}= O_P(1).
\end{equation}
\end{theorem}

\vspace{5mm}

In many practical applications, one does not have not access to the measures $(\mu_i,\nu_i)$.
Instead, one has to make do with observing random samples from each $\mu_i$ and $\nu_i$. In this case, a standard approach is to use smoothed proxies in lieu of the unobservable measures, usually assuming some more regularity. Let $\mu_i^n$ and $\nu_i^n$ be consistent estimators  of $\mu_i$ and $\nu_i$ obtained from smoothing a random sample of size $n$  from each respective measure. Given such estimators, define a new estimator of $T_0$ as
\begin{equation}{\label{estimator-sample}}
 \hat{T}_{n,N}:=\arg\min_{T \in \mathcal{T}_B} \frac{1}{2N}\sum_{i=1}^N d^2_{\mathcal{W}}(T\# \mu_i^n,\nu_i^n),  
\end{equation}
where
$$\mathcal{T}_B:=\{T :\Omega \to \Omega: 0\leq T'(x) < B \text{ for } Q \text{-almost every } x \in \Omega \}\subset \mathcal{T} =\cup_{B>0} \mathcal{T}_B.$$
Note that here one can use any estimators of $\mu_i$ and $\nu_i$ which are consistent in Wasserstein distance, provided $\mu_i^n$ is absolutely continuous. Then, the rate of convergence of $\hat{T}_{n,N}$ will depend on the rate of convergence of  $\mu_i^n$ and $\nu_i^n$ to $\mu_i$ and $\nu_i$, respectively in the Wasserstein distance:

\begin{theorem}{\label{estimator-discrete-samples}}
In the context of model \eqref{model}, suppose that Assumption \ref{absCont} holds true, and furthermore that there exists a $B<\infty$ such that $T_0 \in \mathcal{T}_B$, and $T_\epsilon\in \mathcal{T}_B$ almost surely. Then, the estimator $\hat{T}_{n,N}$ defined in \eqref{estimator-sample} satisfies
\begin{equation}\label{partially-observed-rate}
\norm{\hat{T}_{n,N}-T_0}_{L^2(Q)}\lesssim N^{-1/3}+{r_n}^{-1/2}
\end{equation}
where $r_n^{-1}$ is the rate of convergence in the Wasserstein distance of $\mu_i^n$ to $\mu_i$ and $\nu_i^n$ to $\nu_i$.
\end{theorem}

Precise values of $r_n$ can be obtained by choosing specific estimators and imposing additional regularity on the underlying regressor/response measures.  For instance, one can follow the estimation approach of \cite{weed2019estimation} and obtain the minimax rate of convergence over measures with densities in Besov classes.

\begin{remark}
{Note that $B\in (0,\infty)$ can be any finite constant, however large. Its precise value does not influence the rate \eqref{partially-observed-rate} itself, but only the constants.  It is therefore not to be interpreted as a regularisation parameter. To be strictly faithful to the  assumptions of Theorem \ref{estimator-discrete-samples}, the computation could incorporate additional constraints of the form $(z_{i+1}-z_i)\leq B(x_{i+1}-x_i)$, as a discretization of $T' \leq B$. From a practical point of view, though, we always have $(z_{i+1}-z_i)\leq \big({|\Omega|}/{\min_{1\leq j\leq m} |I_j|}\big)(x_{i+1}-x_i)$, since $T:\Omega\rightarrow\Omega$ is monotone. So maintaining the original formulation of Section \ref{computation} implicitly corresponds to some $B>|\Omega|/\min_{1\leq j\leq m} |I_j|$ in Theorem \ref{estimator-discrete-samples} (recall that $m$ is the user-defined number of nodes in the Riemann sum approximation \eqref{Riemann}). 
}
\end{remark}


\section{Simulated Examples}
In this section we illustrate the estimation framework and finite sample performance of the method by means of some simulations. First we generate random predictors $\{\mu_i\}^N_{i=1}$. We consider random distributions that are mixtures of three independent Beta components. We choose the parameters of the Beta distributions to be uniformly distributed random variables on $[1,10]$, with densities
 $$f_{\mu_i}(x) = \sum_{j=1}^3 \pi_{j}  b_{\alpha_{i,j},\beta_{i,j}}(x), \quad \alpha_{i,j} \sim \text{Uniform}[1,10], \quad \beta_{i,j} \sim \text{Uniform}[1,10].$$
 The $\{\pi_j\}_{j=1}^3$ are arbitrary fixed mixture weights in $[0,1]$, such that $\sum_{j=1}^3 \pi_j=1$.
As for the noise maps $T_{\epsilon_i}$, we use the class of random optimal maps introduced in \citet{panaretos2016amplitude}.  Let $k$ be an integer and define $\zeta_k:[0,1] \to [0,1]$ by $$\zeta_0(x)=x, \quad \zeta_k(x)=x-\frac{\sin(\pi kx)}{|k|\pi}, \qquad k\in Z\setminus \{0\}.$$ These are strictly increasing smooth functions satisfying $\zeta_k(0)=0$ and $\zeta_k(1)=1$ for any $k$. These maps can be made random by replacing $k$ by an integer-valued random variable $K$. If the distribution of $K$ is symmetric around zero, then it is straightforward to see that $E[\zeta_K(x)]=x$, for all  $x\in[0,1]$, as required in the definition of model \eqref{model}. We generate a discrete family of random maps by the following procedure, which is slightly different from the mixture family of maps introduced in \cite{panaretos2016amplitude}: for $J>1$ let $\{K_j\}^J_{j=1}$ be i.i.d. integer-valued symmetric random variables, and $\{U_{(j )}\}^{J-1}_{j=1}$ be the order statistics of $J-1$ i.i.d. uniform random variables on $[0,1]$, independent of $\{K_j\}^J_{j=1}$. The random maps are then defined as
 $$T_{\epsilon}(x)=\sum_{j=1}^{J-1} I(U_{(j)}\leq x \leq U_{(j+1)})\xi(U_{(j)},U_{(j+1)},K_j) (x)$$
 where $\xi(U_{(j)},U_{(j+1)},K_j) (x)$ is defined as the ratio
 $$\Bigg\{\zeta_{K_j}\bigg(\frac{2x}{U_{(j+1)}-U_{(j)}}-\frac{U_{(j+1)}+U_{(j)}}{U_{(j+1)}-U_{(j)}}\bigg)+\frac{U_{(j+1)}+U_{(j)}}{U_{(j+1)}-U_{(j)}}\Bigg\} \bigg/ \Bigg(\frac{2}{U_{(j+1)}-U_{(j)}}\Bigg).$$
As for the optimal map $T_0$ constituting the regression operator, we set $T_0 = \zeta_4$. After having generated the random $\mu_i$ and $T_{\epsilon_i}$, we generate the response distributions according to model \eqref{model}, i.e. $\nu_i = T_{\epsilon_i}\#T\# \mu_i$. Figure (\ref{fig:densities}) depicts representative sample pairs of predictor and response densities.

\begin{figure}[t!]
\centering
 \includegraphics[scale=0.70]{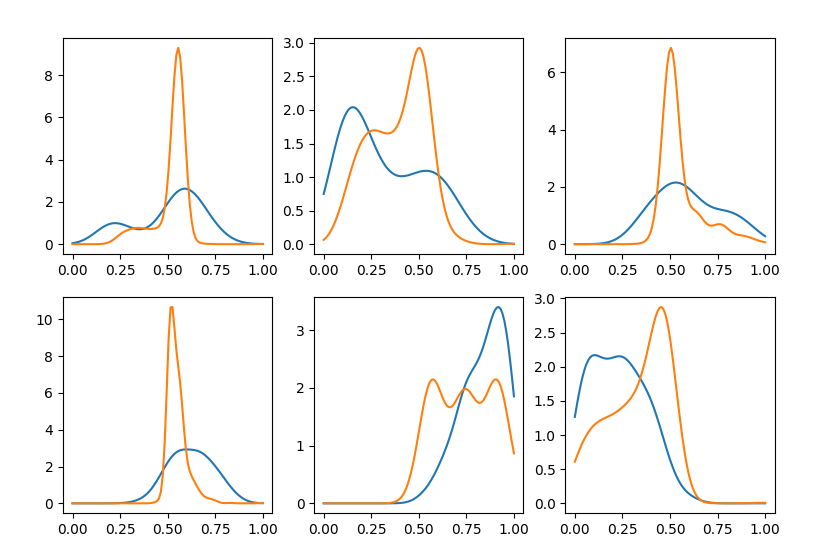}
 \caption{Examples of simulated predictor (blue) and corresponding response (orange) densities.}
 \label{fig:densities}
\end{figure}

\begin{figure}[t!]
\centering
 \includegraphics[scale=0.55]{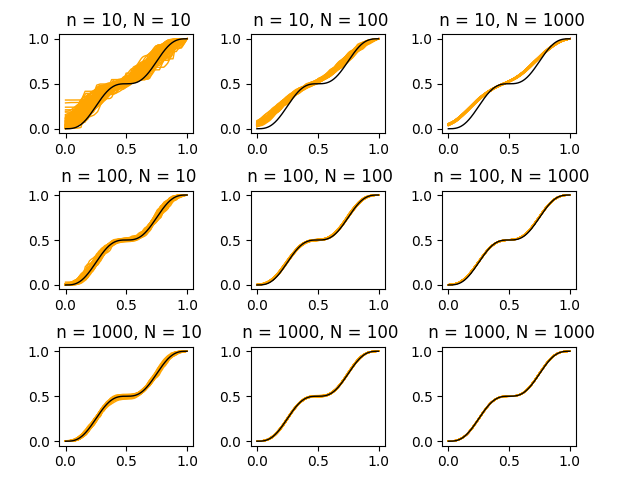}
 \caption{Estimated (yellow) versus true (black) regression map for each of 100 replications of the combinations of $N \in \{10,100,1000\}$ and $n\in \{10,100,1000\}$.}
 \label{fig:nN}
\end{figure}

\begin{figure}[t!]
\centering
 \includegraphics[scale=0.65]{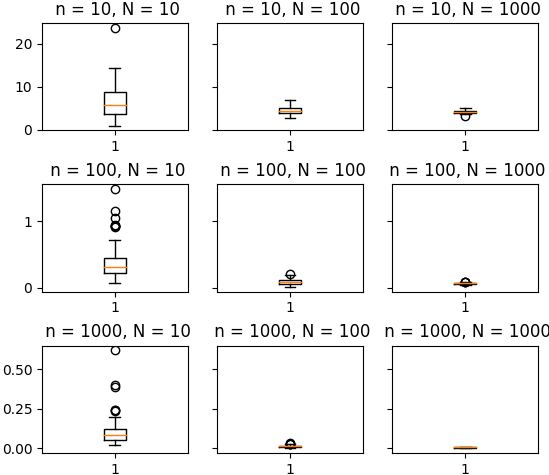}
 \caption{Boxplots for the squared $L^2$ deviation between the  true regression map and the estimated regression maps based on 100 replications for the nine combinations of $N \in \{10,100,1000\}$ and $n\in \{10,100,1000\}$. The $y$-axis scale is common for different values of $N$.}
 \label{fig:nNBoxplot}
\end{figure}

For estimation,  we consider the case where we only observe $n$ independent samples from each pair of distributions $(\mu_i,\nu_i)^N_{i=1}$. For simplicity, we use kernel density estimation, rather than the estimators in \cite{weed2019estimation}, to obtain the proxies $\mu_i^n$ and $\nu_i^n$ for the distributions $\mu_i$ and $\nu_i$. Subsequently, for each $i$, we estimate $Q_i^n$, where $Q_i^n$ is the optimal map such that $\nu_i^n= Q_i\# \mu_i^n$  and solve the convex optimisation problem described in Section \ref{computation} to obtain the estimator $\hat{T}_{n,N}$.  Figure (\ref{fig:nN}) contrasts the estimated and true regression maps in each replication, for all nine combinations $N \in \{10,100,1000\}$ and $n\in \{10,100,1000\}$. It is apparent that the dominant source of error is the bias due to partial observation, i.e. due to observing the measures through finite samples of size $n$. When $n$ is moderately large (e.g. $n=100$) we see that the agreement between estimated and true map is very good, even for small values of $N$. To quantitatively summarise the behaviour of the mean squared error in $N$, we construct boxplots for the error $\|{\hat{T}_{n,N}-T_0}\|_{L^2}$ in Figure (\ref{fig:nNBoxplot}), each based on 100 replications for the corresponding combination of $n\in\{10,100,1000\}$ and $N\in\{10,100,1000\}$. The scale used is the same for each value of $n$, in order to focus the behaviour with respect to $N$.

\section{Analysis of Mortality Data}\label{mortality_data}

We consider the age-at-death distributions for $N=37$ countries in the years 1983 and 2013, obtained from the Human Mortality Database of UC Berkeley and the Max Planck Institute for Demographic Research, openly accessible on \texttt{www.mortality.org}. Death rates are provided by single years of age up to 109, with an open age interval for 110+. We use Gaussian kernel density smoothing, to obtain age-at-death densities from the count data. Denote by $\mu_i$ the age-at-death distribution for the $i$th country at year 1983 and $\nu_i$ the  age-at-death distribution for the same country at year 2013. We use the distributions $\mu_i$ and $\nu_i$ as predictor and response distributions respectively. We chose these two years to allow comparison with \citet{chen2020wasserstein}, who illustrate their methodology on the same data set, and same pair of years.

\begin{figure}[t!]
\centering
\includegraphics[scale = 0.5]{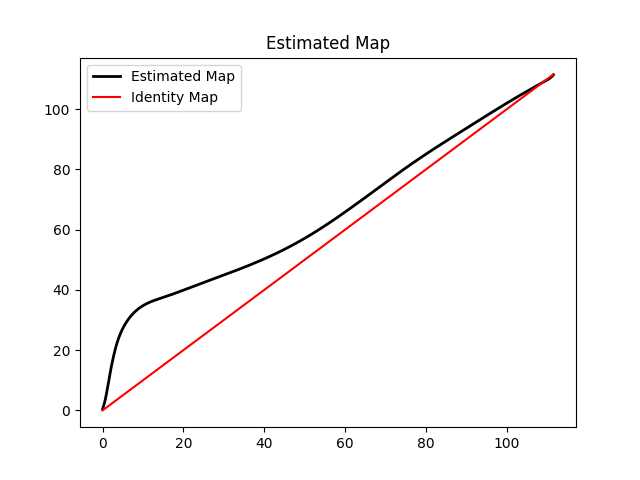}
\caption{Estimated Regression Map for the age-at-death distributional regression (black) contrasted to the identity (red).}
\label{fig:estimatedMap}
\end{figure}

\begin{figure}[t!]
\centering
\includegraphics[scale=0.5]{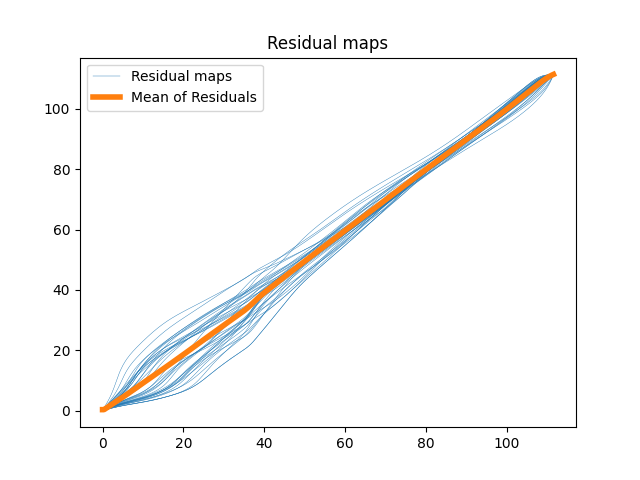}
\caption{Residual maps of all the 37 countries (blue) and their average map (orange).}
\label{fig:diagnostics}
\end{figure}

We fit the model (\ref{model}) by means of the approach described in Section \ref{computation} to obtain the estimated regression map based on the $N=37$ countries. This is depicted in Figure \ref{fig:estimatedMap}. The map dominates the identity map pointwise, indicating that the regression effect is to transport the mass of the age-at-death distribution to the right at visually all locations. Said differently, the map indicates an effect of net improvement in mortality across all ages. The most pronounced such effect is observed in young ages (between 0-10), where the regression map rises steeply: The proportion of the population dying at ages 0-10 in 1983 is redistributed approximately over the range 0-30 in 2013. The form of the map restricted to $[0,10]\mapsto [0,30]$ is approximately linear, indicating that this redistribution is achieved by conserving the actual shape of the distribution but scaling by a constant. The effect is still visible though less pronounced in the early adult to middle age range: The proportion of the population dying at ages between 20 and 60 in 1983 is approximately redistributed over ages 40-60 in 2013. The regression map is approximately parallel to the identity map on the range 60-80, shifted upwards by about 10 years indicating a translation of that interval by that amount of years between 1983 and 2013, i.e. the proportion of the population dying between 60-80 in 1983 has shifted to ages 70-90, but the shape of the distribution of that proportion over each of these two 20 year periods is approximately conserved. Overall, the regression map approximately resembles a piecewise linear map, allowing to interpret it locally by translations and dilations. 

\begin{figure}[t!]
\centering
\includegraphics[width=0.85\textwidth]{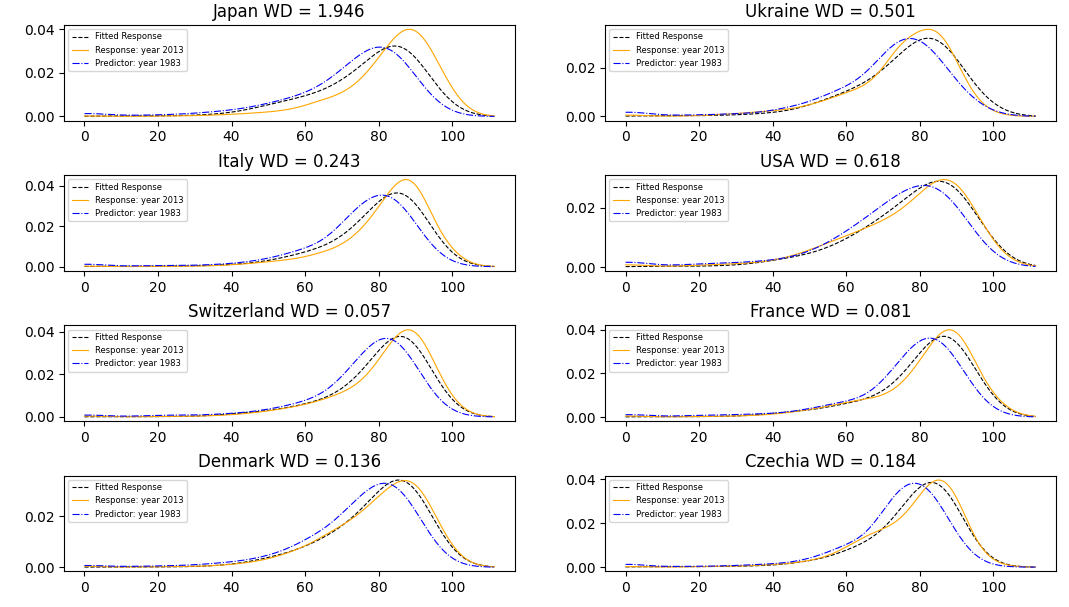}
\caption{Distribution-on-distribution regression for the mortality distributions of Japan, Ukraine, Italy and USA in the year 2013 on those in 1983. Here WD stands for the Wasserstein distance between the observed and fitted densities at year 2013, indicating goodness-of-fit.}
\label{fig:8countries}
\end{figure}

\begin{figure}[t!]
\centering
\includegraphics[width=0.85\textwidth]{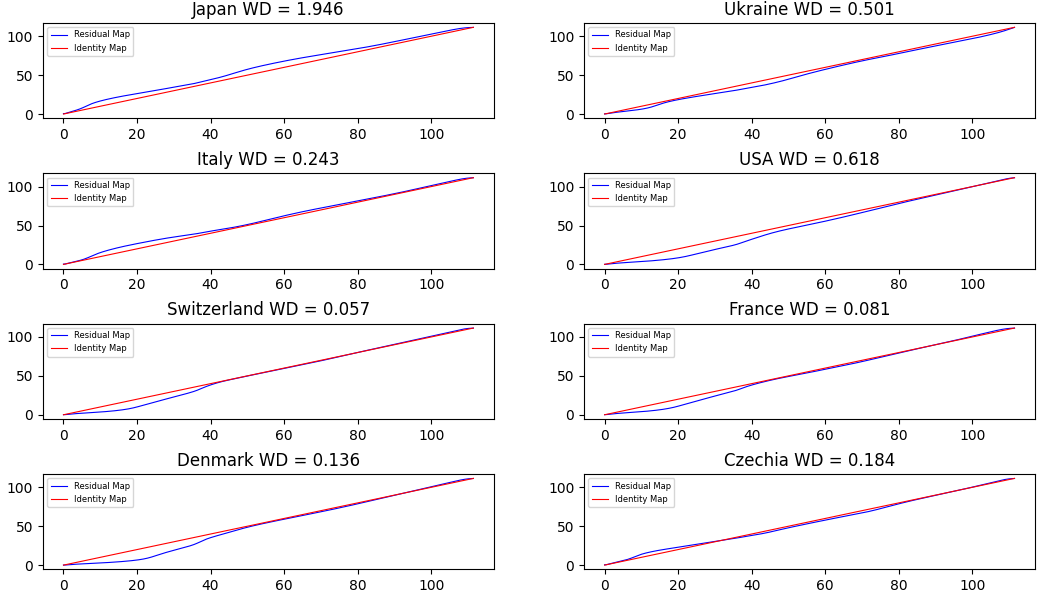}
\caption{Residual maps $T_{\hat{\nu}_i\to\nu_i}$ (blue) vs Identity map (red) for the eight countries in Figure \ref{fig:8countries}.}
\label{fig:8residuals}
\end{figure}

It is not easy to directly compare the effects expressed via this estimated regression map with the effects reflected by the estimated regression coefficient function $\hat\beta$, that is, the integral kernel of the operator $\mathcal{B}$ in Equation \eqref{muller-regressor} obtained in \citet[Figure 3]{chen2020wasserstein}, when fitting their model to the same data. This is largely due to fact that the $\hat{\beta}$ acts on tangent space elements, and thus is rather subtle to interpret. In interpreting their estimated regression operator, those authors remarked that the estimated $\hat{\beta}(s,t)$ was stratified according to the $s$ argument so that, ``\emph{if the log-transformed predictor is non-negative or non-positive throughout its domain, then the fit for the log-transformed response is determined by the comparison of the absolute values of the log transformed predictor over the positive and negative strata of the estimated coefficient $\beta(\cdot,t)$}". 

Using the estimated map $\hat{T}_N$ we can then compute the fitted age-at-death distributions for the year 2013, namely $\hat{\nu}_i=\hat{T}_N\#\mu_i$. {Figure (\ref{fig:8countries}) depicts the predictor and response densities as well as fitted response densities for a sample of 8 different countries. The first four of these countries  (Japan, Ukraine, Italy and USA) were also selected as representative examples in \citet{chen2020wasserstein}.} All eight countries exhibit a negatively-skewed age-at-death distribution. Comparing the actual distributions for the years 1983 and 2013 we can observe the decreasing trend in infant death counts and peaks shifting to older ages, as dictated by the fitted regression map. Contrasting observed and fitted distributions for 2013 allows for better comparison with the model output in \cite{chen2020wasserstein}, than does comparing the estimated regression operators.

Indeed, the main observations made in \cite{chen2020wasserstein} are also apparent from our fitted model. In the case of our model, besides looking at the shape of the predicted densities, we can also take advantage of the direct interpretability of the residual maps $T_{e_i}=T_{\hat{\nu}_i\to \nu_i}$, where $T_{\hat{\nu}_i\to \nu_i}$ is the optimal map between the fitted response $\hat{\nu}_i$ and actual response $\nu_i$. The collection of residual maps is plotted in Figure (\ref{fig:8residuals}). It is apparent that the pointwise variability declines for progressively older ages, illustrating that it is harder to fit mortality at younger ages. One can then focus on the residual maps of specific countries. For example, doing so in the case of Japan and Ukraine, we reproduce the observation in \cite{chen2020wasserstein} that ``\emph{for Japan, the rightward mortality shift is seen to be more expressed than suggested by the fitted model, so that longevity extension is more than is anticipated, while the mortality distribution for Ukraine seems to shift to the right at a slower pace than the fitted model would suggest}". Similarly, we recover the same inference as in \cite{chen2020wasserstein} regarding the US:  ``\emph{while the evolution of the mortality distributions for Japan and Ukraine can be viewed as mainly a rightward shift over calendar years, this is not the case for USA, where compared with the fitted response, the actual rightward shift of the mortality distribution seems to be accelerated for those above age 75} [note: 65 in our case]\emph{, and decelerated for those below age 70 }[note: 65 in our case]". In terms of fit as measured by the Wasserstein distance between response and fit, both models have a harder time fitting Japan, ours doing slightly worse. On the flip side, our model fits Italy better, and the US and Ukraine considerably better (we only contrast countries explicitly mentioned in \cite{chen2020wasserstein}).

Figure \ref{fig:diagnostics} features the overlay of all residual maps, in order to explore the goodness-of-fit of the model as well as the validity of the model assumptions. As the figure shows, the mean of residuals almost matches the identity map, which provides evidence in support of our model specification, in that the residual effects after correcting for the regression should have mean identity, reflected by the assumption that $E\{T_\epsilon(x)\}=x$. Note that, contrary to usual least squares where the residuals have empirical mean zero, the residual maps need not have mean identity exactly.

\begin{figure}[ht!]
\begin{subfigure}{\textwidth}
  \centering
\includegraphics[scale = 1]{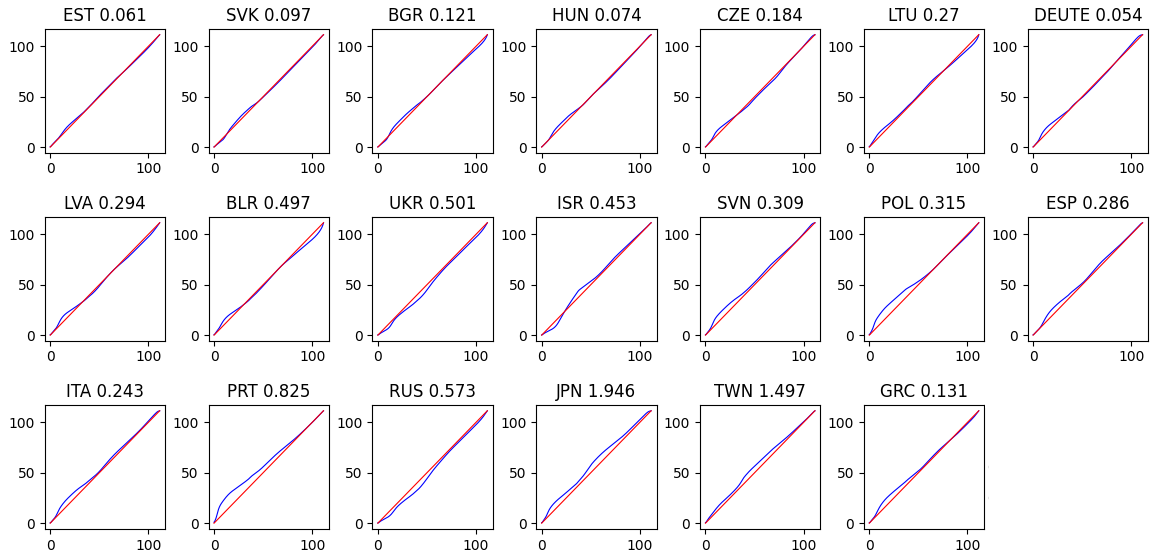}
\caption{}
\label{fig:residuals-east}
\end{subfigure}
\newline
\begin{subfigure}{\textwidth}
  \centering
 \includegraphics[scale = 1]{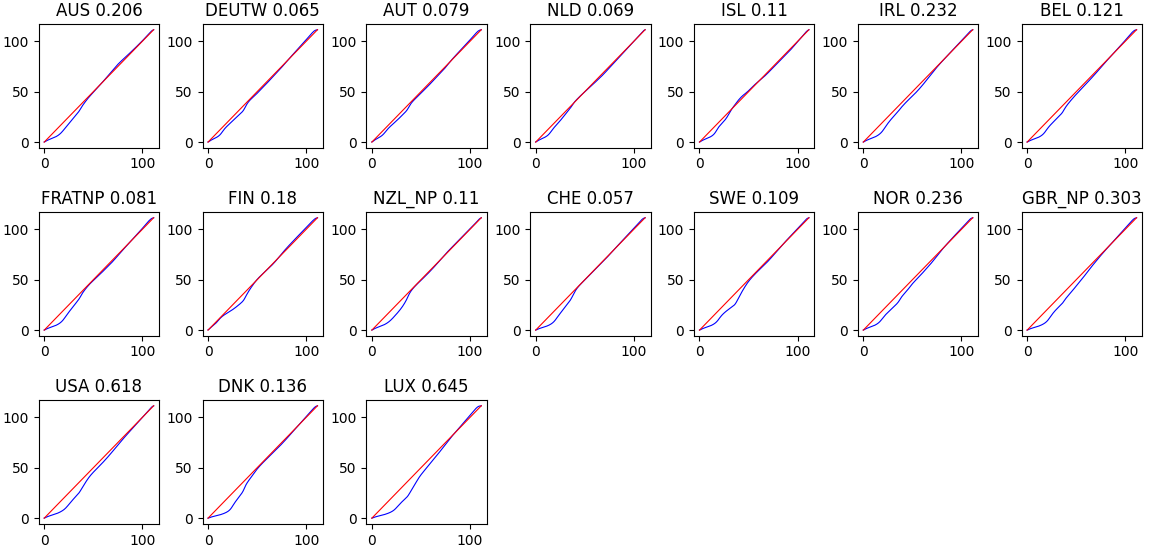}
\caption{}
\label{fig:residuals-west}
\end{subfigure}
\caption{Residual maps (blue), the identity map (red) and the Wasserstein distance between the observed and fitted densities at year 2013 for each country. The countries are clustered in two groups (a) and (b). The list of abbreviations can be found in Table \ref{country-list} in the Supplement.}
\end{figure}

Finally, we can scrutinise the individual residual maps for each of the 37 countries which we plot separately in Figures (\ref{fig:residuals-east}) and (\ref{fig:residuals-west}). The separation into two figures is deliberate, and is based on an apparent clustering: in Figure (\ref{fig:residuals-east}) one can observe more of a rightward shift of fitted mortalities compared to the observed moralities for the countries concerned. This contrasts to countries in Figure (\ref{fig:residuals-west}) which feature less of a rightward shift than fitted by the model. In a sense, these are clusters of ``underfitted" and ``overfitted" observations.  Interestingly countries in Figure (\ref{fig:residuals-east}) belong mostly to Eastern Europe plus Portugal, Spain, Italy, Israel, Japan and Taiwan. Countries in figure (\ref{fig:residuals-west}) belong to western/northern European countries plus USA, New Zealand and Australia. Thanks to the pointwise interpretability of the residual maps, one can notice a particular contrast between these two groups of countries in terms of their fitted/observed infant mortality rates. This may be related to the fact that countries in Figure (\ref{fig:residuals-east}) experienced a more pronounced improvement in their health care systems over the period 1983-2013, compared countries in Figure (\ref{fig:residuals-west}) where healthcare was of comparably high quality already in 1983. It is interesting to note that Japan and Taiwan feature residual maps that everywhere dominate the identity.

\section{Proofs}

\begin{proof}[Proof of Lemma \ref{functional-convexity-derivative}]
Using the closed form of optimal transport maps when $d=1$, one can write:
$$M(T)=\frac{1}{2}\int \int_0^1 \big|T\{F^{-1}_{\mu}(p)\}-F^{-1}_{\nu}(p)\big|^2 \diff p \diff P(\mu,\nu).$$
The expression above shows that $M$ is convex with respect to $T$ since the map $x\to x^2$ is convex and also integration preserves convexity. To show the strict convexity we should prove that for all $0<\beta<1$ and all $T_1,T_2$ such that $\norm{T_1-T_2}^2_{L^2(Q)}>0$,
$$M\big\{\beta T_1 +(1-\beta)T_2\big\} < \beta M(T_1) + (1-\beta) M(T_2).$$
In fact by expanding the squares in the equality and doing some algebra one can conclude that the equality happens if and only if $\norm{T_1-T_2}^2_{L^2(Q)}=0$. Thus $M$, and similarly $M_N$, are strictly convex.

Notice that the domain of definition of $M$ can be extended to the space of $L^2(Q)$ functions. Therefore the Gateaux derivative of $M$ in the direction of $\eta \in L^2(Q)$ can be defined as:
\begin{equation*}
    \begin{split}
        D_\eta M(T) &= \lim_{\epsilon \to 0} \frac{M(T+\epsilon \eta)-M(T)}{\epsilon}.
    \end{split}
\end{equation*}
Expanding the first term we have:
\begin{equation}{\label{perturb}}
    \begin{split}
        M(T+\epsilon \eta) &= M(T) + \epsilon \int \int_0^1 \Big[T\{F^{-1}_{\mu}(p)\}-F^{-1}_{\nu}(p)\Big]\eta\{F^{-1}_{\mu}(p)\}\diff p \diff P(\mu,\nu)\\
        &\quad +\frac{\epsilon^2}{2} \int \int_0^1 \big|\eta\{F^{-1}_{\mu}(p)\}\big| \diff x \diff P(\mu,\nu)\\
        &= M(T) +\epsilon\int <T-F^{-1}_{\nu}\circ F_{\mu},\eta>_{L^2(\mu)}\diff P(\mu,\nu) +\frac{\epsilon^2}{2} \int \norm{\eta}^2_{L^2(\mu)}\diff P(\mu)\\
         &= M(T) +\epsilon\int <T-F^{-1}_{\nu}\circ F_{\mu},\eta>_{L^2(\mu)}\diff P(\mu,\nu) +\frac{\epsilon^2}{2}  \norm{\eta}^2_{L^2(Q)}.
    \end{split}
\end{equation}
The last equality is true since

\begin{equation*}
    \begin{split}
     \int \norm{\eta}^2_{L^2(\mu)} \diff P(\mu) &=\int \int_{\Omega} |\eta(x)|^2 \diff \mu(x) \diff P(\mu,\nu)\\
        &=\int \int_{\Omega} |\eta(x)|^2 \diff Q(x)\\
        &=\norm{\eta}^2_{L^2(Q)}.
    \end{split}
\end{equation*}
Since $\norm{\eta}^2_{L^2(Q)}<\infty$, we can conclude
$$D_{\eta} M(T)=\int <T-F^{-1}_{\nu}\circ F_{\mu},\eta>_{L^2(\mu)}\diff P(\mu,\nu)=\int \int_{\Omega} \{T(x)-T_{\mu,\nu}(x)\}\eta(x)\diff \mu(x)\diff P(\mu,\nu),$$
where $T_{\mu,\nu}$ is the optimal map from $\mu$ to $\nu$. One can use a similar argument to derive the derivative of $M_N$.

\end{proof}

\begin{proof}[Proof of Theorem \ref{identifiability}]

We prove that $T_0$ is the unique minimizer of the population functional in $\mathcal{T}$. Suppose $\nu=T_{\epsilon}\#(T_0\#\mu_0)$ for some fixed measure $\mu_0$, where by assumption $\E\{T_{\epsilon}(x)\}=x$ almost everywhere. Thus according to Proposition 3.2.11 of  \cite{panaretos2020invitation}, $T_0 \# \mu_0$ is the Fr\'echet mean of the conditional probability law of $\nu$ given $\mu_0$ or equivalently, for any $\mu_0$

$${\arg\inf}_{b\in \mathcal{W}_2(\Omega)} \int_{\mathcal{W}_2(\Omega)} d^2_{\mathcal{W}}(b,\nu) \diff P(\nu|\mu_0)=T_0\# \mu_0,$$
where $P$ is the joint distribution of $(\mu,\nu)$ induced by Model \eqref{model}.
Now $T_0$ is a minimizer of the above functional, since for any $T$:
\begin{equation*}
    \begin{split}
        M(T)&=\int d^2_{\mathcal{W}}(T\# \mu,\nu) \diff P(\mu,\nu)\\
        &=\int \int d^2_{\mathcal{W}}(T\#\mu_0,\nu)\diff P(\nu|\mu_0) \diff P(\mu_0)\\
        &\geq \int \int d^2_{\mathcal{W}}(T_0\#\mu_0,\nu)\diff P(\nu|\mu_0) \diff P(\mu_0)\\
        &=\int d^2_{\mathcal{W}}(T_0\# \mu,\nu) \diff P(\mu,\nu).
    \end{split}
\end{equation*}
Also since $d^2_{\mathcal{W}}(T\# \mu,\nu)$ is strictly convex w.r.t. $T \in \mathcal{T}$, and integration preserves strict convexity, the functional $M$ is strictly convex. So $T_0$ is, in fact, the \emph{unique} minimizer.
\end{proof}

\noindent To establish Proposition \ref{Unique-minimizer}, we will use the following theorem.

\begin{theorem}[\citet{kurdila2006convex}, Theorem 7.3.6]{\label{theorem for existence}}
 Let $X$ be a reflexive Banach space and suppose that $f:M\subseteq X \to \mathbb{R}$ is Gateaux-differentiable on the closed, convex and bounded subset $M$. If any of the following three conditions holds true,
\begin{enumerate}
    \item $f$ is convex over $M$,
    \item $Df$ is monotone over $M$,
    \item $D^2 f$ is positive over $M$,
\end{enumerate}
then all three conditions hold, and there exists an $x_0\in X$ such that
$$f(x_0)=\inf_{x\in M} f(x).$$
\end{theorem}

\begin{proof}[Proof of Proposition \ref{Unique-minimizer}]
The set of maps $\mathcal{T}$ is  closed, convex and bounded in the Hilbert space of $L^2(Q)$ functions.
Thus the existence follows immediately from \eqref{functional-convexity-derivative} and Theorem \ref{theorem for existence}.  Uniqueness also follows from strict convexity of $M$.
\end{proof}

To establish the consistency and rate of convergence of our estimator, we will make use of the theory of $M$-estimation. To this aim, we restate some key theorems from \citet{van1996weak}.

\begin{theorem}[\citet{van1996weak}, Theorem 3.2.3]{\label{M-estimation}}
Let $M_n$ be random functions for positive integer $n$, and let $M$ be a fixed function of $\theta$ such that for any $\epsilon>0$
\begin{equation}{\label{uniform convergence}}
    \inf_{d(\theta,\theta_0)\geq \epsilon} M(\theta)> M(\theta_0),
\end{equation}
\begin{equation}{\label{uniqueness-theta}}
   \sup_{\theta}|M_n(\theta)-M(\theta)|\to 0 \quad \text{in probability}.
\end{equation}
Then any sequence of estimators $\hat{\theta}_n$ with $M_n(\hat{\theta}_n)\leq M_n(\theta_0)+o_P(1)$ converges in probability to $\theta_0$.
\end{theorem}

\begin{theorem}[\citet{van1996weak}, Theorem 3.2.5]{\label{theorem3.2.5}}
Let $M_N$ be a stochastic process indexed by a metric space $\Theta$, and let $M$ be a deterministic function, such that for every $\theta$ in a neighborhood of $\theta_0$,
$$M(\theta)-M(\theta_0)\gtrsim d^2(\theta,\theta_0).$$
Suppose that, for every $N$ and sufficiently small $\delta$,

$$\E^* \sup_{d^2(\theta,\theta_0)<\delta} \sqrt{N}\big|(M_N-M)(\theta)-(M_N-M)(\theta_0)\big|\lesssim \phi_N(\delta),$$
for functions $\phi_N$ such that $\delta \to \phi_N(\delta)/\delta^\alpha$ is decreasing for some $\alpha<2$ (not depending on $N$). Let
$$r_N^2\phi_N\left(\frac{1}{r_N}\right)\leq \sqrt{N}, \quad \text{for every } N.$$
If the sequence $\hat{\theta}_N$ satisfies $M_N(\hat{\theta}_N)\leq M_N(\theta_0)+O_P(r_N^{-2})$, and converges in outer probability to $\theta_0$, then $r_N d(\hat{\theta}_N,\theta_0)=O^*_P(1)$. If the displayed conditions are valid for every $\theta$ and $\delta$, then the condition that $\hat{\theta}_N$ is consistent is unnecessary.

\end{theorem}

\begin{theorem}[\citet{van1996weak}, Theorem 2.7.5]{\label{metric-entropy-monotone}}
The class $\mathcal{F}$ of monotone functions $f:\mathbb{R} \to [0,1]$ satisfies
$$\log N_{[]}(\epsilon,\norm{.}_{L^2(Q)},\mathcal{F})\leq K\left(\frac{1}{\epsilon}\right),$$
for every probability measure $Q$, every $p\geq1$, and a constant $K$ that depends only on $p$.
\end{theorem}

\begin{theorem}[\citet{van1996weak}, Theorem 3.4.2]{\label{chaining}}
Let $\mathcal{F}$ be class of measurable functions such that $P f^2 < \delta^2$
and $\norm{f}_{\infty}<M$ for every $f$ in $\mathcal{F}$. Then
$$\E \sup_{f \in \mathcal{F}} |\sqrt{N} (\hat{P}-P)f|\leq \tilde{J}_{[]}(\delta,\norm{.}_{L^2(P)},\mathcal{F})\Bigg(1+\frac{\tilde{J}_{[]}(\delta,\norm{.}_{L^2(P)},\mathcal{F})}{\delta^2 \sqrt{N}} M \Bigg),$$
where $\tilde{J}_{[]}(\delta,\norm{.}_{L^2(P)},\mathcal{F})=\int_0^\delta \sqrt{1+\log N_{[]}(\epsilon,\norm{.}_{L^2(P)},\mathcal{F})} \diff \epsilon$.
\end{theorem}

\begin{proof}[Proof of Theorem \ref{rate of convergence}]

Recall that, from Lemma \ref{Unique-minimizer}, $\hat{T}_N$ is the minimizer of the following criterion  within the function class $\mathcal{T}$:
$$M_N(T):=\frac{1}{2N}\sum_{i=1}^N d^2_{\mathcal{W}}(T\# \mu_i,\nu_i).$$
And the ``true" optimal map $T_0$ is the minimizer of the following criterion function,
$$M(T):=\frac{1}{2}\int d^2_{\mathcal{W}}(T\# \mu,\nu) \diff P(\mu,\nu).$$
First we obtain an adequate upper bound for the bracketing number of the class of functions indexed by $T$ of the form:
$$\mathcal{F}_u:=\{f_T(\mu,\nu) = d^2_{\mathcal{W}}(T\#\mu,\nu)-d^2_{\mathcal{W}}(T_0\#\mu,\nu) , \text{  s.t.  } T\in \mathcal{T} \text{ and } \norm{T-T_0}_{L^2(Q)}\leq u \},$$
where the domain of each function $f_T \in \mathcal{F}_u$ is $\mathcal{W}_2(\Omega) \times \mathcal{W}_2(\Omega)$.
Denote by $\log N_{[]}(\epsilon,\norm{.}_{L^2(Q)},\mathcal{F}_u)$ the bracketing entropy of the function class $\mathcal{F}_u$.
One can directly control this bracketing entropy by the bracketing entropy of the class of optimal maps $\mathcal{T}$ since
\begin{equation}{\label{lipstchitz}}
    \begin{split}
        |d^2_{\mathcal{W}}(T_1\#\mu,\nu)-d^2_{\mathcal{W}}(T_2\#\mu,\nu)| & \leq \norm{T_1-T_2}_{L^2(\mu)}\\
        & \leq C \norm{T_1-T_2}_{L^2(Q)}.
    \end{split}
\end{equation}
Since optimal maps are monotone functions, using Lemma \ref{metric-entropy-monotone}, we know $\log N_{[]}(\epsilon,\norm{.}_{L^2(Q)},\mathcal{T})\leq K\left(\frac{1}{\epsilon}\right)$, and thus we conclude
$$\log N_{[]}(\epsilon,\norm{.}_{L^2(P)},\mathcal{F}_u)\lesssim\left(\frac{1}{\epsilon}\right).$$
The first line of the inequality (\ref{lipstchitz}) also shows that
$$P f_T^2\leq P \norm{T-T_0}^2_{L^2(\mu)}=\norm{T-T_0}^2_{L^2(Q)}\leq u^2,$$
for all $f_T \in \mathcal{F}_u$.

To get the rate of convergence, we first show that $M(T)$ has quadratic growth around its minimizer. For any map $T$, we can write $T = T_0 + \eta $, where $\eta = T - T_0$. Thus the equation (\ref{perturb}), with $\epsilon=1$ and also the fact $D_\eta M(T_0)=0$ yields

\begin{equation*}
    \begin{split}
        M(T)-M(T_0) &= \frac{1}{2} \norm{\eta}^2_{L^2(Q)}\\
        &=\frac{1}{2}\norm{T-T_0}^2_{L^2(Q)}.
    \end{split}
\end{equation*}
Next, we find a function $\phi_N(\delta)$ such that

\begin{equation*}
    \begin{split}
        \E \sup_{\norm{T-T_0}_{L^2(Q)}\leq \delta, T\in \mathcal{T}} \sqrt{N} \Big|(M_{N}-M)(T)-(M_{N}-M)(T_0)\Big|
        &= \E \sup_{f \in F_\delta} \sqrt{N} |(P_N-P)f|\\
        &\leq \phi_N(\delta).
    \end{split}
\end{equation*}
Since the functions in $\mathcal{F}_\delta$ are uniformly bounded and $Pf^2\leq \delta^2$ for all $f\in\mathcal{F}_\delta$, the conditions of Theorem \ref{chaining} are satisfied and we can choose
$$\phi_N(\delta)=\tilde{J}_{[]}(\delta,\norm{.}_{L^2(P)},\mathcal{F}_\delta)\Bigg(1+\frac{\tilde{J}_{[]}(\delta,\norm{.}_{L^2(P)},\mathcal{F}_\delta)}{\delta^2 \sqrt{N}} \bar{c} \Bigg),$$
where the constant $\bar{c}$ is a uniform upper bound for the functions in class $\mathcal{F}_\delta$. Since we noted that $\log N_{[]}(\epsilon,\norm{.}_{L^2(P)},\mathcal{F}_u)\lesssim \epsilon^{-1}$ for any $u>0$, we can show
$$\tilde{J}_{[]}(\delta,\norm{.}_{L^2(P)},\mathcal{F})\leq \int_0^\delta 1+ \sqrt{\log N_{[]}(\epsilon,\norm{.}_{L^2(P)},\mathcal{F}_\delta)} \diff \epsilon \lesssim \sqrt{\delta}.$$
The above inequality and the required condition $\phi_N(\delta)\leq \delta_N^2 \sqrt{N}$ gives the bound $\delta_N=N^{-1/3}$.

\end{proof}

To establish the rate of convergence under imperfect observation we will make use of the following Lemma.

\begin{lemma}{\label{Figalli}}
Let $\mu_n$ be a sequence of measures converging in Wasserstein distance to a measure $\mu$ at a rate of convergence $r_n^{-1}$ and let $T\in \mathcal{T}$. Then $d^2_{\mathcal{W}}(T\# \mu_n,T\#\mu)\lesssim r_n^{-2}$.
\end{lemma}

\begin{proof}
For simplicity and without loss of generality assume that $d^2_{\mathcal{W}}( \mu_n,\mu)=r_n^{-2}$ exactly. If $S_n$ is the optimal map from $\mu_n$ to $\mu$, then
$$
\int \big|S_n(x)-x\big|^2 d\mu_n \leq r_n^{-2}.
$$
Since $T$ is differentiable almost everywhere, and satisfies $|T'(x)| \leq B $ for almost all $x \in \Omega$, then $T$ is Lipschitz continuous with Lipschitz constant at most $B$. Thus

\begin{equation}
    \begin{split}
        d^2_{\mathcal{W}}(T\# \mu_n,T\#\mu)&\leq \int \big|T\{S_n(x)\}-T(x)\big|^2d \mu_n\\
        &\leq B^2 \int \big|S_n(x)-x\big|^{2}d \mu_n   \hspace{10mm} \\
        & \lesssim r_n^{-2}
    \end{split}
\end{equation}

\end{proof}

\begin{proof}[Proof of Theorem \ref{estimator-discrete-samples}]
Define $M_{n,N}(T):=\frac{1}{N}\sum_{i=1}^N d^2_{\mathcal{W}}(T\# \mu_i^n,\nu_i^n)$. 
For any map $T \in \mathcal{T}$,
\begin{equation}{\label{figalli2}}
    \begin{split}
        \E |M_{n,N}(T)-M_N(T)|&=\E\Big|\frac{1}{N}\sum_{i=1}^N d^2_{\mathcal{W}}(T\# \mu_i^n,\nu_i^n)-\frac{1}{N}\sum_{i=1}^N d^2_{\mathcal{W}}(T\# \mu_i,\nu_i)\Big|\\
        &\leq \E\big|d^2_{\mathcal{W}}(T\# \mu_i^n,\nu_i^n)-d^2_{\mathcal{W}}(T\# \mu_i,\nu_i)\big| \\
        &\leq 2C \E\big|d_{\mathcal{W}}(T\# \mu_i^n,\nu_i^n)-d_{\mathcal{W}}(T\# \mu_i^n,\nu_i)\big|+\E\big|d_{\mathcal{W}}(T\# \mu_i^n,\nu_i)-d_{\mathcal{W}}(T\# \mu_i,\nu_i)\big|\\
        &\leq 2C\E d_{\mathcal{W}}( \nu_i^n,\nu_i)+\E d_{\mathcal{W}}(T\# \mu_i^n,T\#\mu_i)\\
        &\lesssim r_n^{-1} \hspace{10mm} \text{ (by Lemma \ref{Figalli})},
    \end{split}
\end{equation}
where $C=\sup_{\mu,\nu} d_{\mathcal{W}}(\mu,\nu)$, and $r_n^{-1}$ is the rate of estimation of an absolutely continuous measure from $n$ samples. Thus the above inequality shows the uniform convergence of $M_{n,N}$ to $M_N$ (at a rate independent of $N$). Also, since $\hat{T}_N$ is the unique minimizer of $M_N$, according to Theorem \ref{M-estimation}, $\hat{T}_{n,N}$ is a consistent estimator for $\hat{T}_N$, when $N$ is fixed.

Now assuming $N$ is fixed, we again use Theorem  \ref{theorem3.2.5} for functionals $M_{n,N}$ and $M_N$. Since both functionals are differentiable, the first condition of the Theorem (quadratic growth) is satisfied. For the second condition we need to find an upper bound for
\begin{equation}
    \begin{split}
      \E \sup_{\norm{T-\hat{T}_N}_{L^2(Q)}<\delta} \sqrt{n} \big|(M_{n,N}-M_N)(T)-(M_{n,N}-M_N)(\hat{T}_{N})\big| &=\phi_n(\delta).
    \end{split}
\end{equation}
According to \eqref{figalli2},
we have $\phi_n(\delta_n)\lesssim r_n^{-1} \sqrt{n}.$
We also need $\phi_n(\delta_n)\leq \sqrt{n} \delta_n^2$, thus $\delta_n^2\sim r_n^{-1}$. Therefore

$$\norm{\hat{T}_{n,N}-\hat{T}_N}_{L^2(Q)}=\delta_n=r_n^{-1/2},$$
and
$$\norm{\hat{T}_{n,N}-T_0}_{L^2(Q)}\leq \norm{\hat{T}_{n,N}-\hat{T}_N}_{L^2(Q)}+\norm{\hat{T}_{N}-T_0}_{L^2(Q)},$$
thus
$$\norm{\hat{T}_{n,N}-T_0}_{L^2(Q)}
\lesssim r_n^{-1/2}+N^{-1/3}.$$

\end{proof}

\begin{table}
\caption{Country abbreviations used in Figures \ref{fig:residuals-east} and \ref{fig:residuals-west}}
\begin{tabular}{ |p{3cm}||p{3cm}|  }
 \hline
 \multicolumn{2}{|c|}{Country List Figure (\ref{fig:residuals-east})} \\
 \hline
 Country Name     & Country Code\\
 \hline
Estonia&EST\\	
Slovakia&SVK\\	
Bulgaria&BGR\\
Hungary&HUN\\
Czechia&CZE\\
Lithuania&LTU\\	
East Germany&DEUTE\\
Latvia&LVA\\
Belarus&BLR\\
Ukraine&UKR\\
Israel&ISR\\
Slovenia&SVN\\	
Poland&POL\\	
Spain&ESP\\	
Italy&ITA\\	
Portugal&PRT\\	
Russia&RUS\\
Japan&JPN\\
Taiwan&TWN\\
Greece&GRC\\
 \hline
\end{tabular}
\quad
\begin{tabular}{ |p{3cm}||p{3cm}|  }
 \hline
 \multicolumn{2}{|c|}{Country List Figure  (\ref{fig:residuals-west})} \\
 \hline
 Country Name     & Country Code\\
 \hline
Australia&AUS	\\	
West Germany&DEUTW\\
Austria&AUT	\\	
Netherlands&NLD\\
Iceland&ISL\\	
Ireland&IRL\\	
Belgium&BEL\\			
France&FRATNP\\				
Finland&FIN\\
New Zealand	&NZL-NP\\
Switzerland&CHE\\	
Sweden&SWE\\	
Norway&NOR\\
United Kingdom&GBR-NP\\	
U.S.A.&USA\\
Denmark&DNK\\			
Luxemburg&LUX\\
\quad&\quad\\
\quad&\quad\\
\quad&\quad\\
 \hline
\end{tabular}
\label{country-list}
\end{table}

\bibliographystyle{imsart-nameyear}
\bibliography{paper}

\begin{thebibliography}{22}

\bibitem[\protect\citeauthoryear{Bigot et~al.}{2018}]{bigot2018upper}
\begin{barticle}[author]
\bauthor{\bsnm{Bigot},~\bfnm{J{\'e}r{\'e}mie}\binits{J.}},
  \bauthor{\bsnm{Gouet},~\bfnm{Ra{\'u}l}\binits{R.}},
  \bauthor{\bsnm{Klein},~\bfnm{Thierry}\binits{T.}},
  \bauthor{\bsnm{Lopez},~\bfnm{Alfredo}\binits{A.}} \betal{et~al.}
(\byear{2018}).
\btitle{Upper and lower risk bounds for estimating the Wasserstein barycenter
  of random measures on the real line}.
\bjournal{Electronic journal of statistics}
\bvolume{12}
\bpages{2253--2289}.
\end{barticle}
\endbibitem

\bibitem[\protect\citeauthoryear{Chen, Lin and
  M{\"u}ller}{2020}]{chen2020wasserstein}
\begin{barticle}[author]
\bauthor{\bsnm{Chen},~\bfnm{Yaqing}\binits{Y.}},
  \bauthor{\bsnm{Lin},~\bfnm{Zhenhua}\binits{Z.}} \AND
  \bauthor{\bsnm{M{\"u}ller},~\bfnm{Hans-Georg}\binits{H.-G.}}
(\byear{2020}).
\btitle{Wasserstein regression}.
\bjournal{arXiv preprint arXiv:2006.09660}.
\end{barticle}
\endbibitem

\bibitem[\protect\citeauthoryear{Delicado}{2011}]{delicado2011dimensionality}
\begin{barticle}[author]
\bauthor{\bsnm{Delicado},~\bfnm{Pedro}\binits{P.}}
(\byear{2011}).
\btitle{Dimensionality reduction when data are density functions}.
\bjournal{Computational Statistics \& Data Analysis}
\bvolume{55}
\bpages{401--420}.
\end{barticle}
\endbibitem

\bibitem[\protect\citeauthoryear{Hall et~al.}{2007}]{hall2007methodology}
\begin{barticle}[author]
\bauthor{\bsnm{Hall},~\bfnm{Peter}\binits{P.}},
  \bauthor{\bsnm{Horowitz},~\bfnm{Joel~L}\binits{J.~L.}} \betal{et~al.}
(\byear{2007}).
\btitle{Methodology and convergence rates for functional linear regression}.
\bjournal{Annals of Statistics}
\bvolume{35}
\bpages{70--91}.
\end{barticle}
\endbibitem

\bibitem[\protect\citeauthoryear{Hsing and Eubank}{2015}]{hsing2015theoretical}
\begin{bbook}[author]
\bauthor{\bsnm{Hsing},~\bfnm{Tailen}\binits{T.}} \AND
  \bauthor{\bsnm{Eubank},~\bfnm{Randall}\binits{R.}}
(\byear{2015}).
\btitle{Theoretical foundations of functional data analysis, with an
  introduction to linear operators}
\bvolume{997}.
\bpublisher{John Wiley \& Sons}.
\end{bbook}
\endbibitem

\bibitem[\protect\citeauthoryear{Kneip and Utikal}{2001}]{kneip2001inference}
\begin{barticle}[author]
\bauthor{\bsnm{Kneip},~\bfnm{Alois}\binits{A.}} \AND
  \bauthor{\bsnm{Utikal},~\bfnm{Klaus~J}\binits{K.~J.}}
(\byear{2001}).
\btitle{Inference for density families using functional principal component
  analysis}.
\bjournal{Journal of the American Statistical Association}
\bvolume{96}
\bpages{519--542}.
\end{barticle}
\endbibitem

\bibitem[\protect\citeauthoryear{Kokoszka
  et~al.}{2019}]{kokoszka2019forecasting}
\begin{barticle}[author]
\bauthor{\bsnm{Kokoszka},~\bfnm{Piotr}\binits{P.}},
  \bauthor{\bsnm{Miao},~\bfnm{Hong}\binits{H.}},
  \bauthor{\bsnm{Petersen},~\bfnm{Alexander}\binits{A.}} \AND
  \bauthor{\bsnm{Shang},~\bfnm{Han~Lin}\binits{H.~L.}}
(\byear{2019}).
\btitle{Forecasting of density functions with an application to cross-sectional
  and intraday returns}.
\bjournal{International Journal of Forecasting}
\bvolume{35}
\bpages{1304--1317}.
\end{barticle}
\endbibitem

\bibitem[\protect\citeauthoryear{Kurdila and
  Zabarankin}{2006}]{kurdila2006convex}
\begin{bbook}[author]
\bauthor{\bsnm{Kurdila},~\bfnm{Andrew~J}\binits{A.~J.}} \AND
  \bauthor{\bsnm{Zabarankin},~\bfnm{Michael}\binits{M.}}
(\byear{2006}).
\btitle{Convex functional analysis}.
\bpublisher{Springer Science \& Business Media}.
\end{bbook}
\endbibitem

\bibitem[\protect\citeauthoryear{Le~Gouic et~al.}{2019}]{gouic2019fast}
\begin{barticle}[author]
\bauthor{\bsnm{Le~Gouic},~\bfnm{Thibaut}\binits{T.}},
  \bauthor{\bsnm{Paris},~\bfnm{Quentin}\binits{Q.}},
  \bauthor{\bsnm{Rigollet},~\bfnm{Philippe}\binits{P.}} \AND
  \bauthor{\bsnm{Stromme},~\bfnm{Austin~J}\binits{A.~J.}}
(\byear{2019}).
\btitle{Fast convergence of empirical barycenters in Alexandrov spaces and the
  Wasserstein space}.
\bjournal{arXiv preprint arXiv:1908.00828}.
\end{barticle}
\endbibitem

\bibitem[\protect\citeauthoryear{Mair, Hornik and
  de~Leeuw}{2009}]{mair2009isotone}
\begin{barticle}[author]
\bauthor{\bsnm{Mair},~\bfnm{Patrick}\binits{P.}},
  \bauthor{\bsnm{Hornik},~\bfnm{Kurt}\binits{K.}} \AND \bauthor{\bparticle{de}
  \bsnm{Leeuw},~\bfnm{Jan}\binits{J.}}
(\byear{2009}).
\btitle{Isotone optimization in R: pool-adjacent-violators algorithm (PAVA) and
  active set methods}.
\bjournal{Journal of statistical software}
\bvolume{32}
\bpages{1--24}.
\end{barticle}
\endbibitem

\bibitem[\protect\citeauthoryear{Morris}{2015}]{morris2015functional}
\begin{barticle}[author]
\bauthor{\bsnm{Morris},~\bfnm{Jeffrey~S}\binits{J.~S.}}
(\byear{2015}).
\btitle{Functional regression}.
\bjournal{Annual Review of Statistics and Its Application}
\bvolume{2}
\bpages{321--359}.
\end{barticle}
\endbibitem

\bibitem[\protect\citeauthoryear{Panaretos and
  Zemel}{2016}]{panaretos2016amplitude}
\begin{barticle}[author]
\bauthor{\bsnm{Panaretos},~\bfnm{Victor~M}\binits{V.~M.}} \AND
  \bauthor{\bsnm{Zemel},~\bfnm{Yoav}\binits{Y.}}
(\byear{2016}).
\btitle{Amplitude and phase variation of point processes}.
\bjournal{The Annals of Statistics}
\bvolume{44}
\bpages{771--812}.
\end{barticle}
\endbibitem

\bibitem[\protect\citeauthoryear{Panaretos and
  Zemel}{2019}]{panaretos2019statistical}
\begin{barticle}[author]
\bauthor{\bsnm{Panaretos},~\bfnm{Victor~M}\binits{V.~M.}} \AND
  \bauthor{\bsnm{Zemel},~\bfnm{Yoav}\binits{Y.}}
(\byear{2019}).
\btitle{Statistical aspects of Wasserstein distances}.
\bjournal{Annual review of statistics and its application}
\bvolume{6}
\bpages{405--431}.
\end{barticle}
\endbibitem

\bibitem[\protect\citeauthoryear{Panaretos and
  Zemel}{2020}]{panaretos2020invitation}
\begin{bbook}[author]
\bauthor{\bsnm{Panaretos},~\bfnm{Victor~M}\binits{V.~M.}} \AND
  \bauthor{\bsnm{Zemel},~\bfnm{Yoav}\binits{Y.}}
(\byear{2020}).
\btitle{An invitation to statistics in Wasserstein space}.
\bpublisher{Springer Nature}.
\end{bbook}
\endbibitem

\bibitem[\protect\citeauthoryear{Patrangenaru and
  Ellingson}{2015}]{patrangenaru2015nonparametric}
\begin{bbook}[author]
\bauthor{\bsnm{Patrangenaru},~\bfnm{Victor}\binits{V.}} \AND
  \bauthor{\bsnm{Ellingson},~\bfnm{Leif}\binits{L.}}
(\byear{2015}).
\btitle{Nonparametric statistics on manifolds and their applications to object
  data analysis}.
\bpublisher{CRC Press}.
\end{bbook}
\endbibitem

\bibitem[\protect\citeauthoryear{Petersen
  et~al.}{2016}]{petersen2016functional}
\begin{barticle}[author]
\bauthor{\bsnm{Petersen},~\bfnm{Alexander}\binits{A.}},
  \bauthor{\bsnm{M{\"u}ller},~\bfnm{Hans-Georg}\binits{H.-G.}} \betal{et~al.}
(\byear{2016}).
\btitle{Functional data analysis for density functions by transformation to a
  Hilbert space}.
\bjournal{Annals of Statistics}
\bvolume{44}
\bpages{183--218}.
\end{barticle}
\endbibitem

\bibitem[\protect\citeauthoryear{Petersen and
  M{\"u}ller}{2019}]{petersen2019frechet}
\begin{barticle}[author]
\bauthor{\bsnm{Petersen},~\bfnm{Alexander}\binits{A.}} \AND
  \bauthor{\bsnm{M{\"u}ller},~\bfnm{Hans-Georg}\binits{H.-G.}}
(\byear{2019}).
\btitle{Fr{\'e}chet regression for random objects with Euclidean predictors}.
\bjournal{Annals of Statistics}
\bvolume{47}
\bpages{691--719}.
\end{barticle}
\endbibitem

\bibitem[\protect\citeauthoryear{Petersen, Zhang and
  Kokozska}{2021+}]{petersen-review}
\begin{barticle}[author]
\bauthor{\bsnm{Petersen},~\bfnm{Alexander}\binits{A.}},
  \bauthor{\bsnm{Zhang},~\bfnm{Chao}\binits{C.}} \AND
  \bauthor{\bsnm{Kokozska},~\bfnm{Piotr}\binits{P.}}
(\byear{2021}+).
\btitle{Modeling Probability Density Functions as Data Objects}.
\bjournal{Econometrics and Statistics}
\bpages{(to appear)}.
\end{barticle}
\endbibitem

\bibitem[\protect\citeauthoryear{Van Der~Vaart and Wellner}{1996}]{van1996weak}
\begin{bincollection}[author]
\bauthor{\bsnm{Van Der~Vaart},~\bfnm{Aad~W}\binits{A.~W.}} \AND
  \bauthor{\bsnm{Wellner},~\bfnm{Jon~A}\binits{J.~A.}}
(\byear{1996}).
\btitle{Weak convergence}.
In \bbooktitle{Weak convergence and empirical processes}
\bpages{16--28}.
\bpublisher{Springer}.
\end{bincollection}
\endbibitem

\bibitem[\protect\citeauthoryear{Weed and Berthet}{2019}]{weed2019estimation}
\begin{barticle}[author]
\bauthor{\bsnm{Weed},~\bfnm{Jonathan}\binits{J.}} \AND
  \bauthor{\bsnm{Berthet},~\bfnm{Quentin}\binits{Q.}}
(\byear{2019}).
\btitle{Estimation of smooth densities in Wasserstein distance}.
\bjournal{arXiv preprint arXiv:1902.01778}.
\end{barticle}
\endbibitem

\bibitem[\protect\citeauthoryear{Zemel and Panaretos}{2019}]{zemel2019frechet}
\begin{barticle}[author]
\bauthor{\bsnm{Zemel},~\bfnm{Yoav}\binits{Y.}} \AND
  \bauthor{\bsnm{Panaretos},~\bfnm{Victor~M.}\binits{V.~M.}}
(\byear{2019}).
\btitle{Fr{\'e}chet means and Procrustes analysis in Wasserstein space}.
\bjournal{Bernoulli}
\bvolume{25}
\bpages{932--976}.
\end{barticle}
\endbibitem

\bibitem[\protect\citeauthoryear{Zhang, Kokoszka and
  Petersen}{2020}]{zhang2020wasserstein}
\begin{barticle}[author]
\bauthor{\bsnm{Zhang},~\bfnm{Chao}\binits{C.}},
  \bauthor{\bsnm{Kokoszka},~\bfnm{Piotr}\binits{P.}} \AND
  \bauthor{\bsnm{Petersen},~\bfnm{Alexander}\binits{A.}}
(\byear{2020}).
\btitle{Wasserstein Autoregressive Models for Density Time Series}.
\bjournal{arXiv preprint arXiv:2006.12640}.
\end{barticle}
\endbibitem

\end{thebibliography}

\end{document}